\documentclass{amsart}
\usepackage{amssymb}

\newtheorem{theorem}{Theorem}[section]
\newtheorem{lemma}[theorem]{Lemma}
\newtheorem{corollary}[theorem]{Corollary}

\theoremstyle{definition}

\newtheorem{example}[theorem]{Example}

\theoremstyle{remark}
\newtheorem{remark}[theorem]{Remark}

\numberwithin{equation}{section}

\newcommand{\F}{\mathbb{F}}

\newcommand{\Mon}{\mathrm{Mon}}

\newcommand{\Aut}{\mathrm{Aut}}

\newcommand{\lcm}{\mathrm{lcm}}
\newcommand{\Fsig}{ F_{\sigma}(C)}
\newcommand{\Esig}{ E_{\sigma}(C)}

\begin{document}

\title{On the Structure of the Linear Codes with a Given Automorphism}


\author{Stefka Bouyuklieva}
\address{Faculty of Mathematics and Informatics, St. Cyril and St. Methodius University of Veliko Tarnovo, BULGARIA}
\email{stefka@ts.uni-vt.bg}
\thanks{The research is partially supported by the Bulgarian National Science Fund under Contract No KP-06-H62/2/13.12.2022}


\subjclass[2010]{Primary 94B05,20B25}

\keywords{linear codes, automorphisms, quasi-cyclic codes}

\date{}

\dedicatory{}

\begin{abstract}
The purpose of this paper is to present the structure of the linear codes over a finite field with $q$ elements that have a permutation automorphism of order $m$. These codes can be considered as generalized quasi-cyclic codes. Quasi-cyclic codes and almost quasi-cyclic codes are discussed in detail, presenting necessary and sufficient conditions for which linear codes with such an automorphism are self-orthogonal, self-dual, or linear complementary dual.
\end{abstract}

\maketitle

\section{Introduction}\label{sect:introduction}

Linear codes invariant under a given permutation have been studied for a long time, and this is most evident for cyclic codes (we refer to \cite{MacWilliams_Sloane_1977} and \cite{Handbook} for more information), and also for group codes (see \cite{Berman,group_codes}). The idea of using automorphisms in the construction of combinatorial structures is not new. In \cite{Anstee}, the authors used automorphisms in the search for a projective plane of order 10. The research of W. C. Huffman is devoted to the study of linear and more precisely self-dual codes over various finite fields and even rings having an automorphism of a given order (preferably prime) \cite{Huffman48,Huffman4I,Huffman3,Huffman2013}. We would like to mention also the works of V. Yorgov, G. Nebe,  M. Borello and W. Willems on the linear codes and their automorphisms (for example \cite{Borello_2p,Nebe,Yorgov56}).

Studying the structure of linear codes invariant under a given permutation is important not only from the point of view of obtaining useful information about the properties and parameters of these codes, but also for presenting efficient methods for constructing codes with given parameters and preset properties such as self-orthogonal or linear complementary dual (LCD) codes. Codes that have such an automorphism, have some symmetric structure and useful algebraic properties. As a disadvantage of these methods, we can point out the lack of comprehensiveness, i.e. we are not sure whether we have obtained all codes with the requested properties and parameters, we cannot prove the nonexistence of codes with given parameters unless we combine the research with other techniques. On the other hand, codes with a large group of automorphisms have a rich algebraic structure, very useful properties, practical applications, and therefore they are the most studied codes.

In this paper, we present a study on linear codes over a finite field with $q$ elements ($q$ is a prime power) having as an automorphism a permutation $\sigma$ of a given order $m$ (not necessarily prime), focusing in particular on the case where $\sigma$ has $c$ disjoint cycles of length $m$ and $f\ge 0$ fixed points. Our idea is to combine:
 \begin{itemize}
 \item[(1)] the works by Huffman (see for example \cite{Huffman48,Huffman3}) and Yorgov \cite{Yorgov56}, mostly on binary, ternary and Hermitian quaternary self-dual codes,
 \item[(2)] our own research on binary self-dual, self-orthogonal and LCD codes having an automorphism of prime order \cite{DCC97-order2,Stef2000,StefkaJavier,BMW2010}, and
 \item[(3)] the algebraic approach to quasi-cyclic codes of Patrick Sol\'{e} and San Ling \cite{LSole,LSoleIV}.
 \end{itemize}
 We present some general conditions according to which the considered linear codes are self-dual, self-orthogonal, or LCD, respectively.

The paper is organized as follows. In the next section we present the needed definitions and general statements. Section \ref{sect:perm-coprime} is devoted to permutations of order $m$, relatively prime to the characteristic of the considered finite field. The term \emph{almost quasi-cyclic code} is introduced for generalized quasi-cyclic (GQC) codes of block lengths $(m,\ldots,m,1,\ldots,1)$. These are linear codes invariant under a permutation of order $m$ with $c$ disjoint cycles of length $m$ and $f\ge 1$ cycles of length 1 (fixed points) in its decomposition. In Section \ref{sect:perm-notcoprime} we study codes with a permutation automorphism of order $m$ divisible by the characteristic of the field. We present some examples in Section \ref{sect:exam}. We end the paper with a conclusion.

\section{Preliminaries}
\label{Preliminaries}

A linear $[n,k]$ code $C$ is a $k$-dimensional subspace of the
vector space $\mathbb{F}_q^n$, where $\mathbb{F}_q$ is the finite
field of $q$ elements, $q=p^{\ell}$ for a prime $p$ and positive integer $\ell$.
Let $(u,v): \mathbb{F}_q^n\times \mathbb{F}_q^n\to \mathbb{F}_q$
be an inner product in $\mathbb{F}_q^n$.
If $C$ is an $[n,k]$ linear code, then its orthogonal complement
$C^{\perp}=\{u \in \mathbb{F}_q^n : (u,v)=0$ $\forall v \in C \}$ is a linear $[n,n-k]$ code called the dual code of $C$. We consider three types of linear codes depending on the intersection with their duals:
\begin{itemize}
\item If $C = C^{\perp}$, $C$ is termed self-dual. If the length of a self-dual code is $n$ then its dimension must be $n/2$.
\item If $C \subseteq C^{\perp}$, the code is self-orthogonal. A self-orthogonal code is also self-dual iff its dimension is a half of its length. Self-orthogonal codes with $k>n/2$ do not exist.
\item If $C\cap C^\perp$ consists only of the zero vector, the code is called LCD (linear complementary dual). If $C$ is an LCD code so is its dual code $C^\perp$.
    \end{itemize}

In this paper, we consider inner products of two types:
\begin{itemize}
\item Euclidean inner product, defined by
$$u\cdot v=\sum_{i=1}^n u_iv_i\in\F_q, \; u=(u_1,\ldots,u_n), v=(v_1,\ldots,v_n)\in\F_q^n.$$
\item Hermitian inner product
$$(u,v)=\sum_{i=1}^n u_i\overline{v}_i\in\F_q,$$
where $\overline{a}=a^{\sqrt{q}}$ if $q$ is a square.
\end{itemize}

The most general definition for equivalence of linear codes of length $n$ over the finite field $\mathbb{F}_q$ is based on the action of the semilinear isometries group $\mathcal{M}_n^*(q)=\Mon_n(\F_q^\ast)\rtimes {\rm Aut}(\F_q)\leq \Gamma_n(\F_q)$ on the vector space $\F_q^n$, where $\Gamma_n(\F_q)$ is the set of all semilinear mappings, i.e. the general semilinear group, $\Mon_n(\F_q^\ast)$ is the group of all monomial $n\times n$ matrices over $\F_q$, and  ${\rm Aut}(\F_q)$ is the automorphisms group of the field $\F_q$.
Linear $q$-ary codes $C$ and $C^{ \prime}$ of the same length $n$ are equivalent
whenever $C'=CT$ for some $T\in \mathcal{M}_n^*(q)$. If $CT=C$ for an element $T\in \mathcal{M}_n^*(q)$ then $T$ is called an automorphism of the code. The set of all automorphisms of $C$ form a group denoted by $\textrm{Aut}(C)$.

Any element $T\in \mathcal{M}_n^*(q)$ can be written as $T=PD\tau$ where $P$ is a permutation
matrix (permutation part), $D$ is a diagonal matrix (diagonal part), and $\tau\in \Aut(\F_q)$. Note that in the case of prime $q$, $\mathcal{M}_n^*(q)=\Mon_n(\F_q^\ast)$, and if $q=2$ then $\mathcal{M}_n^*(q)\cong \textrm{Sym}(n)$ where $\textrm{Sym}(n)$ is the symmetric group of degree $n$. We consider here only the permutation automorphisms of a linear code.

Let $C$ be a linear $q$-ary code with a permutation automorphism $\sigma\in \textrm{Sym}(n)$ of order $m$ (not necessarily prime). If $\sigma$ is a product of $s$ disjoint cycles, namely
\begin{equation}\label{sigma}
\sigma=\Omega_1\Omega_2\cdots\Omega_s,
\end{equation}
where the length of $\Omega_i$ is $l_i\ge 1$, $1\le i\le s$, then $m=\lcm(l_1,\ldots,l_s)$.

To describe the structure of the considered codes, we need the factor rings $\mathcal{R}_{l_i}=\F_q[x]/(x^{l_i}-1)$, where $\F_q[x]$ is the ring of polynomials in the indeterminate $x$ with coefficients in $\F_q$.
Define the map $\phi: \F_q^n\to \mathcal{R}_{l_1}\times \mathcal{R}_{l_2}\times\cdots\times \mathcal{R}_{l_s}$ by
\begin{align*}
  \phi(c) & =(c_1(x),\ldots,c_s(x)) \\
   & =(c_{10}+c_{11}x+\cdots+c_{1,l_1-1}x^{l_1-1},\ldots,c_{s0}+c_{s1}x+\cdots+c_{s,l_s-1}x^{l_s-1})
\end{align*}
for $c=(c_{10},c_{11},\ldots,c_{1,l_1-1},\ldots,c_{s0},c_{s1},\ldots,c_{s,l_s-1})\in \F_q^n$.

Any submodule of $\mathcal{R}'=\mathcal{R}_{l_1}\times \mathcal{R}_{l_2}\times\cdots\times \mathcal{R}_{l_s}$ is called a generalized quasi-cyclic (GQC) code of block lengths $(l_1,\ldots,l_s)$, which is a linear code of length $l_1+\cdots+l_s$ over $\F_q$ \cite{GQC,Siap}. Decomposition into constituents for GQC codes is given by Esmaeili and Yari in \cite{Yari}. We will not consider this decomposition here.

%
%
%

We define the subcodes of $C$  $$F_{\sigma}(C):=\lbrace v\in C\mid v\sigma=v \rbrace$$ and
$$E_{\sigma}(C):=\{ v=(v_1,\ldots,v_n)\in C: \sum_{i\in \Omega_j}v_i=0 \; \mbox{in} \; \F_q \; \textrm{for all}\; j=1, \ldots, s \}.$$
Note that $v\in \Fsig$ if and only if $v\in C$ and $v|_{\Omega_j}$ is constant for $j=1,\ldots,s$. Therefore, we define the map $\pi: \Fsig\rightarrow\F_q^{s}$ by $(\pi(v))_j=v_i$ for some $ i \in \Omega_j$, $j=1, 2, \ldots, s$, $v\in F_\sigma(C)$.

We use also the map $\psi:C\to\F_q^s$ defined by
$$\psi(v)=(\sum_{i\in\Omega_1} v_i,\ldots,\sum_{i\in\Omega_s} v_i),$$
where $v_i$ are the coordinates of the vector $v\in\F_q^n$, $i=1,\ldots,n$. This is a homomorphism and the kernel of $\psi$ is the subcode $\Esig$, or
$\ker \psi=\Esig$.

Instead of studying linear codes over rings that are Cartesian products of different factor rings, we prefer to focus on some specific cases. Sylow's first theorem gives us a reason to narrow down the considered cases. More precisely, we use the corollary, also known as Cauchy's theorem.

\begin{theorem}{\rm (Cauchy)}\label{thm:Sylow}
Given a finite group $G$ and a prime number $r$ dividing the order of $G$, then there exists an element (and thus a cyclic subgroup generated by this element) of order $r$ in $G$.
\end{theorem}

This means that if a linear code has a nontrivial automorphism group, it has an automorphism of prime order. Therefore, if we need to classify all codes with given parameters, having a nontrivial automorphism group, it is sufficient to consider only the automorphisms of prime orders. However, one can get many more results by combining automorphisms and from this point of view study codes with automorphisms of composite order or codes invariant under the action of a given group (see for example \cite{Borello_Nebe,Borello_2p,Nebe}, as well as the examples in Section \ref{sect:exam}).

In this work, we focus on one automorphism of a linear code without considering its interaction with other automorphisms. The following lemma gives us a partial motivation to consider only permutational automorphisms, especially when their order is prime.

\begin{lemma}{\rm\cite[W. C. Huffman]{Huffman3}}\label{lemma:perm}
Let $C$ be a linear code over $\F_q$ with an automorphism $T = PD\tau$ of prime order $r$ where $r\nmid (q - 1)$ and $r\nmid |{\rm Gal}(\F_q)|$. Then there exists a code $C'$ equivalent to $C$ where $P\in\mathrm{Aut}(C')$.
\end{lemma}

By $\mathbf{1}$ and $\mathbf{0}$ we denote the all-ones vectors and the zero vector of the corresponding length, respectively.

\section{Permutation automorphisms of order $m$ relatively prime to the characteristic of the field}
\label{sect:perm-coprime}

Let $\gcd (m,\textrm{char}(\F_q))=1$. The following theorem gives a very important information about the structure of a linear code $C$ having a permutation automorphism of order $m$.

 \begin{theorem}\label{codecomp}
 Let $C \leq \F_ q^n$ be a linear code with a permutation automorphism $\sigma \in \textrm{Sym}(n)$ of order $m$ such that $\gcd(m,q)=1$. Then $C=F_\sigma(C) \oplus E_\sigma(C)$. Both $F_\sigma(C)$ and $E_\sigma(C)$ are $\sigma$-invariant and orthogonal to each other.
 \end{theorem}

 \begin{proof} Take an arbitrary codeword $v\in C$ and consider $w=\sum_{i=0}^{m-1}v\sigma^i$. Since $v\sigma^m=v$, $w\sigma=w$ and so $w\in F_{\sigma}(C)$. Let $x=v-\frac{1}{m}w$. If $v\vert_{\Omega_i}=(v_{i1},\ldots,v_{il_i})$ then $w\vert_{\Omega_i}=\frac{m}{l_i}(\sum_{j=1}^{l_i}v_{ij},\ldots,\sum_{j=1}^{l_i}v_{ij})$. Hence
 $$x\vert_{\Omega_i}=(v_{i1}-\frac{1}{l_i}\sum_{j=1}^{l_i}v_{ij},\ldots,v_{il_i}-\frac{1}{l_i}\sum_{j=1}^{l_i}v_{ij}) \; \Longrightarrow
 \; \sum_{j=1}^{l_i}x_j=\sum_{j=1}^{l_i}v_{ij}-\sum_{j=1}^{l_i}v_{ij}=0.$$
 It follows that $x\in E_{\sigma}(C)$ and so $v=\frac{1}{m}w+x\in F_{\sigma}(C)+E_{\sigma}(C)$. This proves that $C=F_{\sigma}(C)+E_{\sigma}(C)$.

 If $v\in F_{\sigma}(C)\cap E_{\sigma}(C)$, then $v\vert_{\Omega_i}=(\underbrace{\alpha,\ldots,\alpha}_{l_i})=\alpha\mathbf{1}$ and $\alpha(\underbrace{1+\cdots+1}_{l_i})=l_i\alpha=0$. Hence $\alpha=0$, $v=\mathbf{0}$, $F_{\sigma}(C)\cap E_{\sigma}(C)=\{\mathbf{0}\}$ and therefore $C=F_\sigma(C) \oplus E_\sigma(C)$. Obviously, both subcodes are $\sigma$-invariant.

 If $v\in E_{\sigma}(C)$ and $w\in F_{\sigma}(C)$, $w\vert_{\Omega_i}=w_i\mathbf{1}$ for $i=1,\ldots,s$, then
 $$(v,w)=\sum_{i=1}^s (w_i'\sum_{j\in \Omega_i}v_j)=0,$$ where $w_i'=w_i$ in the case of Euclidean inner product, and $w_i'=\overline{w}_i$ if the inner product is Hermitian. Hence both subcodes are orthogonal to each other.
 \end{proof}

  The following theorem proves important properties of the projection code $C_{\pi}=\pi(F_{\sigma}(C))$ if $C$ is a self-dual or LCD code with respect to the considered inner product.

  \begin{theorem}\label{thm:Fsig}
  Assume $l_1\equiv l_2\equiv \cdots\equiv l_s\equiv l\not\equiv 0\pmod{p}$. Then:
  \begin{itemize}
  \item[(1)] if $C$ is self-orthogonal so is $C_{\pi}$;
   \item[(2)] if $C$ is self-dual so is $C_{\pi}$;
    \item[(3)] if $C$ is LCD so is $C_{\pi}$.
  \end{itemize}
  \end{theorem}

  \begin{proof}
   If $v=(v_1,\ldots,v_n)$, by $v'$ we denote the vector $v$ if the considered inner product is Euclidean, and the vector $\overline{v}=(\overline{v}_1,\ldots,\overline{v}_n)$, if the inner product is Hermitian.

   Let $v=(\underbrace{v_1,\ldots,v_1}_{l_1},\ldots,\underbrace{v_s,\ldots,v_s}_{l_s})$ and $w=(\underbrace{w_1,\ldots,w_1}_{l_1},\ldots,\underbrace{w_s,\ldots,w_s}_{l_s})$ be codewords in $\Fsig$. Then
   $$(v,w)=\sum_{i=1}^s l_iv_iw_i'=l\sum_{i=1}^s v_iw_i'=l(\pi(v),\pi(w)).$$
   \begin{itemize}
   \item[(1)] If $C$ is a self-orthogonal code, then $(v,w)=0$ for any two codewords $v,w\in\Fsig$. Hence $(\pi(v),\pi(w))=0$ $\forall v,w\in\Fsig$ and therefore $C_{\pi}$ is also self-orthogonal.
\item[(2)] Let $C$ be a self-dual code. Hence $C$ is also self-orthogonal and therefore $C_{\pi}$ is a self-orthogonal code, or $C_{\pi}\subseteq C_{\pi}^{\perp}$. Take $w=(w_1,\ldots,w_s)\in C_{\pi}^{\perp}$, and $w_F=(\underbrace{w_1,\ldots,w_1}_{l_1},\ldots,\underbrace{w_s,\ldots,w_s}_{l_s})\in\F_q^n$. Then
    $$(v,w_F)=l(\pi(v),w)=0 \; \forall v\in\Fsig.$$
    Furthermore, $$(u,w_F)=\sum_{i=1}^s (w'_i\sum_{j\in\Omega_i} u_j)=0 \; \forall u\in\Esig.$$
    Hence, $w_F\perp C$ and so $w_F\in C^{\perp}=C$. It follows that $w_F\in\Fsig$ and $w\in C_{\pi}$. This proves that $C_{\pi}^\perp= C_{\pi}$ and so $C_{\pi}$ is a self-dual code.
\item[(3)] Consider now the case of an LCD code $C$. Take $w=(w_1,\ldots,w_s)\in C_{\pi}\cap C_{\pi}^{\perp}$, and $w_F=\pi^{-1}(w)=(\underbrace{w_1,\ldots,w_1}_{l_1},\ldots,\underbrace{w_s,\ldots,w_s}_{l_s})$. Then $w_F\in\Fsig$ and
    $$(v,w_F)=l(\pi(v),w)=0 \; \forall v\in\Fsig.$$
    Hence $w_F\perp \Fsig$ and $w_F\perp \Esig$, which means that $w_F\in C^\perp\cap C$. Since $C$ is an LCD code, $w_F=\mathbf{0}$ and so $w$ is the zero vector.
   It follows that $C_{\pi}\cap C_{\pi}^{\perp}=\{\mathbf{0}\}$ and $C_{\pi}$ is an LCD code.
   \end{itemize}
Thus the theorem is proved.
  \end{proof}

 Next, we focus on quasi-cyclic and almost quasi-cyclic codes.

\subsection{Quasi-cyclic codes}

This is the case when $l_1=l_2=\ldots=l_s=m$ and $n=sm$. Then for the fixed subcode we have the following corollary, that follows from Theorem \ref{thm:Fsig}.

 \begin{corollary}\label{thm:FsigQS}
  Let $C$ be a $q$-ary quasi-cyclic code of length $sm$ and index $s$, where\\ $\gcd(m,q)=1$. Then:
  \begin{itemize}
  \item[(1)] if $C$ is self-orthogonal so is $C_{\pi}$;
   \item[(2)] if $C$ is self-dual so is $C_{\pi}$;
    \item[(3)] if $C$ is LCD so is $C_{\pi}$.
  \end{itemize}
  \end{corollary}

   If $m$ and $q$ are relatively prime, so $p\nmid m$, $x^m-1$ can be written in the form \cite{LSole}
   \begin{equation}\label{xm-1}
x^m-1=\delta g_0(x)g_1(x)\cdots g_r(x)h_1(x)h_1^*(x)\cdots h_t(x)h_t^*(x),
\end{equation}
where $\delta\in\F_q^*$, $g_0=x-1$, $g_1,\ldots,g_r$ are associated with their reciprocal polynomials, and $h_i^*(x)$ is the reciprocal polynomial of $h(x)$, $i=1,\ldots,t$. Then
$$\mathcal{R}_m=(\bigoplus_{i=0}^r \F_q[x]/(g_i))\oplus (\bigoplus_{i=1}^t (\F_q[x]/(h_i)\oplus \F_q[x]/(h^*_i)),$$
and $\F_q[x]/(g_i)$, $i=0,1,\ldots,r$, $\F_q[x]/(h_i)$ and $\F_q[x]/(h^*_i)$, $j=1,\ldots,t$, are fields (extensions of $\F_q$). In some cases it is more suitable to consider these fields as minimal ideals in $\mathcal{R}_m$, generated respectively by the polynomials $\frac{x^m-1}{g_i(x)}$, $i=0,1,\ldots,r$, $\frac{x^m-1}{h_j(x)}$ and $\frac{x^m-1}{h_j^*(x)}$, $j=1,\ldots,t$. Denote these ideals by $G_0,G_1,\ldots,G_r$, $H_1'$, $H_1''$, $\ldots, H_t'$, $H_t''$, respectively, and so
$$\mathcal{R}_m=(\bigoplus_{i=0}^r G_i)\oplus (\bigoplus_{i=1}^t (H_j'\oplus H_j'').$$

In this case the map $\phi$ is defined in the following way:
$$\phi:\F_q^{ms}\to\mathcal{R}_m^s, \; \phi(c)=(c_1(x),\ldots,c_s(x))$$

\begin{lemma}{\rm \cite[Lemma 3.1]{LSole}}
The map $\phi$ induces a one-to-one correspondence
between quasi-cyclic codes over $\F_q$ of index $s$ and length $ms$
and linear codes over $\mathcal{R}_m$ of length $s$.
\end{lemma}

Since $$\mathcal{R}_m^s=(\bigoplus_{i=0}^r G_i^s)\oplus (\bigoplus_{i=1}^t (H_j'^s\oplus H_j''^s),$$
then
$$\phi(C)=(\bigoplus_{i=0}^r C_i)\oplus (\bigoplus_{i=1}^t (C_j'\oplus C_j''),$$
where $C_i$ is a linear code over $G_i$, $i=0,1,\ldots,r$, $C_j'$ and $C_j''$ are linear codes over $H_j'$ and $H_j''$, respectively, $j=1,\ldots,t$, all of length $s$.

Since $G_0=(1+x+\cdots+x^{m-1})\lhd \mathcal{R}_m$, $G_0\cong \F_q[x]/(x-1)\cong\F_q$, $\phi^{-1}(C_0)$ is actually the fixed subcode $\Fsig$, and $$\phi(\Esig)=(\bigoplus_{i=1}^r C_i)\oplus (\bigoplus_{i=1}^t (C_j'\oplus C_j'').$$

In $\mathcal{R}_m^s$, we use the Hermitian inner product, defined in \cite{LSole}, namely
\begin{equation}\label{Hermitian}
    (u,v)=\sum_{i=1}^s u_i \overline{v_i} \ \ \mbox{for} \ u=(u_1,\ldots, u_s), \ v=(v_1, \ldots, u_s),
\end{equation}
where $\overline{v_i}=v_i(x^{-1}) =v_i(x^{m-1}).$
Note that $\overline{v_i}\in G_j$ if $v_i\in G_j$, $0\le j\le r$,  $\overline{v_i}\in H_j''$ if $v_i\in H_j'$, and $\overline{v_i}\in H_j'$ if $v_i\in H_j''$, $1\le j\le t$. Actually, this inner product is Euclidean over $G_0\cong\F_q$.

The following theorem follows from \cite[Theorem 4.2]{LSole}.

\begin{theorem} The linear code $C$ is Euclidean self-dual $q$-ary code if and only if
$C_{\pi}=C_0$ is Euclidean self-dual, $C_i$ for $i=1,\ldots,r$ are Hermitian self-dual codes, and $C_j''=(C_j')^\perp$ for $j=1,\ldots,t$ with respect to the Euclidean inner product.
\end{theorem}

This theorem gives us the following corollaries.

\begin{corollary}\label{cor:QC-SO}
The linear code $C$ is Euclidean self-orthogonal $q$-ary code if and only if
$C_{\pi}=C_0$ is Euclidean self-orthogonal, $C_i$ for $i=1,\ldots,r$ are Hermitian self-orthogonal codes, and $C_j''\subseteq(C_j')^\perp$ for $j=1,\ldots,t$ with respect to the Euclidean inner product.
\end{corollary}

\begin{corollary}\label{cor:QC-LCD}
The linear code $C$ is Euclidean LCD $q$-ary code if and only if
$C_{\pi}=C_0$ is Euclidean LCD code, $C_i$ for $i=1,\ldots,r$ are Hermitian LCD codes, $C_j''\cap (C_j')^\perp=\{\mathbf{0}\}$ and $C_j'\cap (C_j'')^\perp=\{\mathbf{0}\}$ for $j=1,\ldots,t$ with respect to the Euclidean inner product.
\end{corollary}

\subsection{Almost quasi-cyclic codes}

Let $C$ be a linear code with a permutation automorphism $\sigma\in \textrm{Sym}(n)$ of order $m$ (not necessarily prime) with $c$ cycles of length $m$ and $f$ fixed points. We call such code \textit{almost quasi-cyclic}. In this case, we say that $\sigma$ is of type $m$-$(c,f)$.  Without loss of generality we can assume that
\begin{equation}\label{forma-auto}
\sigma =\Omega_1\dots\Omega_c\Omega_{c+1}\dots\Omega_{c+f}
\end{equation}
where $\Omega_{i}=((i-1)m+1,\dots,im), i=1,\dots,c$, are the
cycles of length $r$, and $\Omega_{c+i}=(cm+i), i=1,\dots,f$, are
the fixed points. Obviously, $cm+f=n$. Almost quasi-cyclic codes are a special case of generalized quasi-cyclic codes, but we consider them separately because the decomposition given above is not very useful, since the constituents that correspond to the fixed points are codes of length 1. Therefore, in this subsection we focus on the subcodes $\Fsig$ and $\Esig$ in more detail here. Theorem \ref{thm:Fsig} gives us the following statement.

\begin{corollary}\label{cor:FsigQS}
Let $m\equiv 1\pmod{p}$. Then:
  \begin{itemize}
  \item[(1)] if $C$ is self-orthogonal so is $C_{\pi}$;
   \item[(2)] if $C$ is self-dual so is $C_{\pi}$;
    \item[(3)] if $C$ is LCD so is $C_{\pi}$.
  \end{itemize}
\end{corollary}

If $v\in\Esig$ then $v\vert_{\Omega_j}=0$ for $j=c+1,\ldots,c+f$. Denote by $E_{\sigma}(C)^{*}$ the code obtained from $E_{\sigma}(C)$ by
deleting the last $f$ coordinates. Since $E_{\sigma}(C)^*$ is a quasi-cyclic $q$-ary code of length $cm$ and index $c$, we can use the decomposition given in the previous subsection. All codewords of $C$ that $\sigma$ preserves belong to the subcode $\Fsig$, therefore for the code $C_{\phi}=\phi(E_{\sigma}(C)^{*})$ we have
$$C_{\phi}=(\bigoplus_{i=1}^r C_i)\oplus (\bigoplus_{i=1}^t (C_j'\oplus C_j''),$$
where $C_i$ is a linear code over $G_i$, $i=1,\ldots,r$, $C_j'$ and $C_j''$ are linear codes over $H_j'$ and $H_j''$, respectively, $j=1,\ldots,t$, all of length $c$. Theorem \ref{codecomp} gives us the following corollary.

\begin{corollary}
The code $C$ having an automorphism $\sigma$ given in \eqref{forma-auto} is self-orthogonal (resp. LCD) code if and only if $\Fsig$ and $\Esig^*$ are self-orthogonal (resp. LCD).
\end{corollary}

\begin{proof}
Obviously, $E_{\sigma}(C)$ is an orthogonal (resp. LCD) code if and only if the code $E_{\sigma}(C)^*$ is self-orthogonal (resp. LCD).

If $C$ is a self-orthogonal code, all its subcodes are also self-orthogonal. Conversely, if $\Fsig$ and $\Esig^*$ are self-orthogonal codes then $\Esig$ is also self-orthogonal, and since $\Fsig\perp\Esig$, the code $C=\Fsig\oplus\Esig$ is self-orthogonal.

In the case of LCD codes, we will prove that if $C=C_1\oplus C_2$ and $C_1\perp C_2$, then $C$ is an LCD code if and only if both $C_1$ and $C_2$ are LCD codes.

$\Rightarrow)$ Let $C$ be an LCD code. If $w=(w_1,\ldots,w_{n})\in C_1\cap C_1^\perp$ then $w\perp C_1$ and $w\perp C_2$. This gives us that $w\perp C$ and so $w\in C\cap C^\perp$. Hence $w=0$ and $C_1$ is an LCD code. The same holds for the code $C_2$.

$\Leftarrow)$ Let $C_1$ and $C_2$ be LCD codes, and $x\in C\cap C^\perp$. Since $C= C_1\oplus C_2$ then $x=x_1+x_2$, $x_i\in C_i$, $i=1,2$. Take $y_i\in C_i$, $i=1,2$. Then we have $x\cdot y_i=0$ and
$$x_i\cdot y_i=(x_1+x_2)\cdot y_i=x\cdot y_i=0 \ \Rightarrow x_i\perp C_i \ \Rightarrow x_i\in C_i\cap C_i^\perp \ \Rightarrow x_i=0, \
i=1,2.$$
This proves that $x=0$ and so $C$ is also an LCD code.

To complete the proof, we take $C_1=E_\sigma(C)$ and $C_2=F_\sigma(C)$.
\end{proof}

Combining the corollary with Corollary \ref{cor:QC-SO} and Corollary \ref{cor:QC-LCD}, we prove the following.

\begin{corollary}\label{cor:AQC-SO}
The linear code $C$ is Euclidean self-orthogonal $q$-ary code if and only if
$C_{\pi}$ is Euclidean self-orthogonal, $C_i$ for $i=1,\ldots,r$ are Hermitian self-orthogonal codes, and $C_j''\subseteq(C_j')^\perp$ for $j=1,\ldots,t$ with respect to the Euclidean inner product.
\end{corollary}

\begin{corollary}\label{cor:AQC-LCD}
The linear code $C$ is Euclidean LCD $q$-ary code if and only if
$C_{\pi}$ is Euclidean LCD code, $C_i$ for $i=1,\ldots,r$ are Hermitian LCD codes, $C_j''\cap (C_j')^\perp=\{\mathbf{0}\}$ and $C_j'\cap (C_j'')^\perp=\{\mathbf{0}\}$ for $j=1,\ldots,t$ with respect to the Euclidean inner product.
\end{corollary}


Self-dual almost quasi-cyclic codes was studied by Huffman in \cite{Huffman-decomposing}. Methods for construction and classification of self-dual codes with an automorphism of prime order $m\neq p$ were given in \cite{Huffman48,Yorgov56} for binary codes, \cite{Huffman3} for ternary codes, \cite{Huffman4I,Huffman4II} for Hermitian quaternary codes. More detailed list with references on linear codes with automorphism of prime order can be seen in \cite{Huffman2013}. Binary LCD codes having an automorphism of odd prime order are studied in detail in \cite {StefkaJavier}. Some classes of quasi-cyclic codes with complementary duals are examined in \cite{LCD_Yari}.

\section{Permutation automorphisms of order $m$ not coprime with the characteristic of the field}
\label{sect:perm-notcoprime}

Let $C$ be a linear code over $\F_q$, where $q=p^{\ell}$ for a prime $p$, $\ell\ge 1$, with a permutation automorphism $\sigma\in \textrm{Sym}(n)$ of order $m=p^a m'$, $a\ge 1$. We again consider $\sigma$ as a product of $s$ disjoint cycles as in \eqref{sigma}.

Let us see what happens to the subcodes $\Fsig$ and $\Esig$ in this situation.

\begin{theorem}\label{thm:FsigQS1}
  Let $l_i\equiv 0\pmod{p}$ for all $i=1,\ldots,s$. Then $\Fsig$ is a self-orthogonal code and it is a subcode of $\Esig$.
  \end{theorem}

  This theorem shows that we cannot use the same decomposition as in Section \ref{sect:perm-coprime} in order to study self-orthogonal, self-dual and/or LCD codes having an automorphism $\sigma$. We need something different here.

  Next, we prove a theorem that holds for the quasi-cyclic codes of length $ms$ and index $s$ for all integers $m=p^a m'$. A similar theorem for binary self-dual codes of length $2k$ and index $k$ is proved in \cite{DCC97-order2}.

  \begin{theorem}\label{thm:psi}
  Let $C$ be a $q$-ary quasi-cyclic code of length $sm$ and index $s$. Let $\psi:C\to\F_q^s$ be the map defined by
$$\psi(c)=(\sum_{i\in\Omega_1} c_i,\ldots,\sum_{i\in\Omega_s} c_i),$$
where $c_i$ are the coordinates of the vector $c\in\F_q^n$, $i=1,\ldots,n$. If $C$ is a self-orthogonal code then $C_{\psi}=\psi(C)$ is also self-orthogonal, and $C_{\pi}=\pi(\Fsig)\subset C_{\psi}^{\perp}$. If $C$ is self-dual then $C_{\pi}=C_{\psi}^{\perp}$.
  \end{theorem}

  \begin{proof}
  Let $c=(c_{11},\ldots,c_{1m},c_{21},\ldots,c_{2m},\ldots,c_{s1},\ldots,c_{sm})\in C$. Then $$c\sigma=(c_{1m},c_{11},\ldots,c_{1,m-1},c_{2m},c_{21},\ldots,c_{2,m-1},\ldots,c_{sm},c_{s1},\ldots,C_{s,m-1})\in C.$$

  For two codewords $u,v\in C$ we have
  $$(\psi(u),\psi(v))=\sum_{i=1}^s (u_{i1}+\cdots+u_{im})(v_{i1}'+\cdots+v_{im}')= (u,v)+(u,v\sigma)+\cdots+(u,v\sigma^{m-1}).$$
If $C$ is a self-orthogonal code, then $(u,v)=(u,v\sigma)=\cdots=(u,v\sigma^{m-1})=0 \forall u,v\in C$. Hence $(\psi(u),\psi(v))=0$ for all $u,v\in C$.

Furthermore, $(\psi(u),\pi(v))=\sum_{i=1}^s (u_{i1}+\cdots+u_{im})v_{i}'= (u,v)$. Therefore, for a self-orthogonal code $C$, $(\psi(u),\pi(v))=0$ for $u\in C$, $v\in\Fsig$ and so $C_{\pi}\subset C_{\psi}^{\perp}$.

Now let $C$ be a self-dual code, $v=(v_1,\ldots,v_s)\in C_{\psi}^{\perp}$ and $$w=\pi^{-1}(v)=(\underbrace{v_1,\ldots,v_1}_m,\ldots,\underbrace{v_s,\ldots,v_s}_m).$$ It follows that
$$(u,w)=\sum_{i=1}^s (u_{i1}+\cdots+u_{im})v_{i}'=(\psi(u),v)=0 \; \forall u\in C.$$
Hence, $w\in C^\perp=C$ and so $w\in\Fsig$. This proves that $v\in C_{\pi}$ and therefore $C_{\pi}=C_{\psi}^\perp$.
  \end{proof}

  \begin{remark}\rm
  Theorem \ref{thm:FsigQS1} holds for all $m\ge 2$ and for all prime powers $q$ (if corresponding quasi-cyclic codes exist). If $\gcd(m,q)=1$ then the codes $C_{\psi}$ and $C_{\pi}$ coincide.
  \end{remark}

Quasi-cyclic codes of length $ms$ and index $s$ in the case $m=p^am'$, $a\ge 1$, $p\nmid m'$, are extensively studied in \cite{LSoleIV}. The factorization of the polynomial $x^m-1$ over $\F_q$ plays a key role in this case, too. Since $\gcd(m',q)=1$, the polynomial $x^{m'}-1$ can be factorized as in \eqref{xm-1}, namely
$$x^{m'}-1=\delta g_0(x)g_1(x)\cdots g_r(x)h_1(x)h_1^*(x)\cdots h_t(x)h_t^*(x).$$

Since $x^m-1=x^{p^am'}-1=(x^{m'}-1)^{p^a}$, we have
$$x^m-1=\delta^{p^a} g_0^{p^a}g_1^{p^a}\cdots g_r^{p^a}h_1^{p^a}(h_1^*)^{p^a}\cdots h_t^{p^a}(h_t^*)^{p^a},$$
where $\delta\in\F_q^*$, $g_0=x-1$, $g_1,\ldots,g_r$ are associated with their reciprocal polynomials, and $h_i^*(x)$ is the reciprocal polynomial of $h(x)$, $i=1,\ldots,t$. Consequently, we may now write
$$\mathcal{R}_m=(\bigoplus_{i=0}^r \F_q[x]/(g_i^{p^a}))\oplus (\bigoplus_{i=1}^t (\F_q[x]/(h_i^{p^a})\oplus \F_q[x]/((h^*_i)^{p^a})).$$

If we denote these factor rings by $R_i$ for $i=0,1,\ldots,r$,  $R_j'$ and $R_j''$ for $j=1,\ldots,t$, respectively, then
$$\mathcal{R}_m^s=(\bigoplus_{i=0}^r R_i^s)\oplus (\bigoplus_{i=1}^t (R_j'^s\oplus R_j''^s)).$$

In particular, $\mathcal{R}_m$-linear code $\mathcal{A}$ of length $s$ can be decomposed in a direct sum in the following way
$$\mathcal{A}=(\bigoplus_{i=0}^r \mathcal{A}_i)\oplus (\bigoplus_{i=1}^t (\mathcal{A}'_j\oplus \mathcal{A}''_j)),$$
where $\mathcal{A}_i$, $\mathcal{A}_j'$ and $\mathcal{A}_j''$ are linear codes over the rings $R_i$, $R_j'$ and $R_j''$, respectively, $i=0,1,\ldots,r$, $j=1,\ldots,t$. The rings $R_i$, $R_j'$ and $R_j''$ are finite chain rings. This can be described in the following way (see \cite{LSoleIV}): If $f$ is a monic irreducible factor of $x^{m'}-1$ of degree $d$ then the factor ring $\F_q[x]/(f^{p^a})$ can be identified with the finite chain ring $\F_{q^k}+u\F_{q^k}+\cdots+u^{p^a-1}\F_{q^k}$, where $u^{p^a}=0$. The detailed description of quasi-cyclic codes in this case, as well as some interesting examples, are presented in \cite{LSoleIV}, so we will not consider this theory in more detail, but for completeness we present \cite[Theorem 4.2]{LSoleIV}.

\begin{theorem}\label{QS_autp}
A linear code $C$ over $\mathcal{R}_m=\F_q[x]/(x^m-1)$ of length $s$ is self-dual
with respect to the Hermitian inner product (or equivalently, an $s$-quasi-cyclic
code of length $sm$ over $\F_q$ is self-dual with respect to the Euclidean
inner product) if and only if
$$C=(\bigoplus_{i=0}^r C_i)\oplus (\bigoplus_{i=1}^t (C'_j\oplus (C'_j)^\perp)),$$
where, for $0\le i\le r$, $C_i$ is a self-dual code over $R_i$ of length $s$ (with respect
to the Hermitian inner product) and, for $1\le j\le t$, $C_j'$ is a linear code
of length $s$ over $R_j'$ and $(C'_j)^\perp$ is its dual with respect to the Euclidean
inner product.
\end{theorem}

  The binary self-dual codes invariant under a permutation $\sigma$ of order 2 with $c$ independent cycles of length 2 and $f>0$ fixed points are studied in \cite{Stef2000}. A construction method for such codes is proposed which is used to obtain optimal self-dual codes of different lengths. In \cite{Borello_Nebe}, the authors prove that the natural projection of the fixed code of an involution of a self-dual binary linear code is self-dual under some (quite strong) conditions on the codes. To prove that, they introduce the family of binary semi self-dual codes.

\section{Examples}
\label{sect:exam}

In this section, we give three examples of codes constructed using the presented decomposition. We consider an almost quasi-cyclic binary self-dual code with an automorphism of order 15, a quasi-cyclic binary self-dual code with an automorphism of order 10, and an almost quasi-cyclic ternary LCD code with an automorphism of order 5.

\begin{example}\rm
In the first example we present a binary self-dual almost quasi-cyclic code $C$ with two 15-cycles and two fixed points. Now $C_{\pi}$ must be a binary self-dual $[4,2,2]$ code and we take $C_\pi=\{0000,1010,0101,1111\}$. The code $\Esig^*$ is a binary quasi-cyclic code of length 30 and index 2, which is self-orthogonal code of dimension 14. Since
$$x^{15}-1=(x-1)(x^2+x+1)(x^4+x^3+x^2+x+1)(x^4+x+1)(x^4+x^3+1),$$
where $g_1(x)=x^2+x+1$ and $g_2(x)=x^4+x^3+x^2+x+1$ are self-reciprocal polynomials, and $h(x)=x^4+x+1$ and $h^*(x)=x^4+x^3+1$ are mutually reciprocal polynomials over $\F_2$. It follows that
$$\mathcal{R}_{15}=\F_2[x]/(g_0)\oplus \F_2[x]/(g_1)\oplus \F_2[x]/(g_2)\oplus \F_2[x]/(h)\oplus \F_2[x]/(h^*).$$
Instead of the factor-rings in the above formula, we use the corresponding ideals of ${\mathcal R}_m$:
$$G_1=\langle\frac{x^{15}-1}{g_1(x)}\rangle\cong\F_4, \; G_2=\langle\frac{x^{15}-1}{g_2(x)}\rangle\cong\F_{16},$$ $$H'=\langle\frac{x^{15}-1}{h(x)}\rangle\cong\F_{16}, \; H''=\langle\frac{x^{15}-1}{h^*(x)}\rangle\cong\F_{16}.$$
The generating idempotents of these fields are $e_1(x)=x^{14}+x^{13}+x^{11}+x^{10}+x^8+x^7+x^5+x^4+x^2+x\in G_1$,
$e_2(x)=x^{14}+x^{13}+x^{12}+x^{11}+x^{9}+x^8+x^7+x^6+x^4+x^3+x^2+x\in G_2$, $e'(x)=x^{12}+x^9+x^8+x^6+x^4+x^3+x^2+x\in H'$ and $e''(x)=e'(x^{-1})=x^{14}+x^{13}+x^{12}+x^{11}+x^9+x^7+x^6+x^3\in H''$.
For the code $C_{\phi}$ we have
$$C_{\phi}=C_1\oplus C_2\oplus C'\oplus C'',$$
where $C_i$ is a Hermitian self-dual code over $G_i$, $i=1,2$, $C'$ and $C''$ are mutually orthogonal linear codes over $H'$ and $H''$, respectively, with respect to the Euclidean inner product. We take $C_i=\langle (e_i(x),e_i(x))\rangle$, $i=1,2$, $C'=(H')^2$, and so $C''$ is the zero code. The constructed binary code is doubly-even self-dual $[32,16,8]$ code with a generator matrix $gen_1$, and its weight enumerator is $1+620y^{8}+13888y^{12}+36518y^{16}+13888y^{20}+620y^{24}+y^{32}$.
\[gen_1=\footnotesize{\begin{pmatrix}
01101101101101101101101101101100\\
10110110110110110110110110110100\\
01111011110111101111011110111100\\
10111101111011110111101111011100\\
11011110111101111011110111101100\\
11101111011110111101111011110100\\
01111010110010000000000000000000\\
00111101011001000000000000000000\\
00011110101100100000000000000000\\
10001111010110000000000000000000\\
00000000000000001111010110010000\\
00000000000000000111101011001000\\
00000000000000000011110101100100\\
00000000000000010001111010110000\\
11111111111111100000000000000010\\
00000000000000011111111111111101
\end{pmatrix}}\]

\end{example}

\begin{example}\rm
We consider again a binary self-dual code $C$ but this time we use a permutation automorphism $\sigma$ with four 10-cycles. Since
$$x^{10}-1=(x-1)^2(x^4+x^3+x^2+x+1)^2$$
over $\F_2$, for the ring $\mathcal{R}_{10}$ we have
$$\mathcal{R}_{10}=\F_2[x]/((x-1)^2)\oplus \F_2[x]/((x^4+x^3+x^2+x+1)^2).$$
According to Theorem \ref{QS_autp}, $\phi(C)=C_0\oplus C_1$, where $C_i$ is a linear Hermitian self-dual code over $G_i$, $i=0,1$, $G_0=<(x^4+x^3+x^2+x+1)^2>\cong \F_2[x]/((x-1)^2)$, $G_1=<(x-1)^2>\cong \F_2[x]/((x^4+x^3+x^2+x+1)^2)$. The structure of the rings $G_0$ and $G_1$ is as follows:
$$G_0=\{0,e=1+x^2+x^4+x^6+x^8,u=1+x+\cdots+x^9,e+u\}=\F_2+u\F_2,$$
$$G_1=\F_{16}+u'\F_{16}, \; u'=1+x^4+x^5+x^9, \; \F_{16}^*=\{\beta^i: \; \beta=1+x^2, \; i=0,\ldots,14\}.$$
The identity element of $G_1$ is $e'=x^2+x^4+x^6+x^8=1+e$. The Euclidean and Hermitian inner products in $G_0$ are the same, as $e(x^{-1})=e(x)$ and $u(x^{-1})=u(x)$. Taking the codes
$$C_0=\langle\begin{pmatrix}
e&0&e&u\\
0&e&u&e
\end{pmatrix}\rangle \ \ \ \mbox{and} \ \ \
C_1=\langle\begin{pmatrix}
e'&e'&u'&0\\
0&xu'&e'&e'\\
u'&u'&0&0\\
0&0&u'&u'
\end{pmatrix}
\rangle,$$
we obtain a binary quasi-cyclic self-dual $[40,20,8]$ doubly-even code whose automorphism group has order 245760. The code $C_0$ is the only Type II code over $\F_2+u\F_2$ up to equivalence \cite{F2uF2}.
\end{example}

The third example presents a ternary LCD code.

\begin{example}\rm
If $C$ is a ternary LCD code of length 18 having an automorphism of order 5 with three 5-cycles and three fixed points, its subcodes $\Fsig$ and $\Esig$ are also LCD codes, $C_{\pi}$ is a ternary code of length 6, and $C_{\phi}$ is an LCD code of length 3 with respect to the Hermitian inner product over the field $\F_3[x]/(x^4+x^3+x^2+x+1)$ with $3^4$ elements. This field is described in detail in \cite{Huffman3}, where Huffman has used it to construct ternary self-dual codes.

We take $C_{\pi}=\langle\begin{pmatrix} 110011\\ 001110\end{pmatrix}\rangle$ and $C_{\phi}=\langle\begin{pmatrix} e&e&0\\ 0&\alpha&e\end{pmatrix}\rangle$, where $e=2+x+x^2+x^3+x^4$, $\alpha=x^3+2x^4$. The constructed $[18,10]$ code has minimum distance 4 and a generator matrix $gen_3$.
\[gen_3=\footnotesize{\begin{pmatrix}
111111111100000&011\\
000000000011111&110\\
211112111100000&000\\
121111211100000&000\\
112111121100000&000\\
111211112100000&000\\
000000001221111&000\\
000002000112111&000\\
000001200011211&000\\
000000120011121&000
\end{pmatrix}} \ \ \ \ gen_4=\footnotesize{\begin{pmatrix}
111111111100000&011\\
000000000011111&111\\
211112111100000&000\\
121111211100000&000\\
112111121100000&000\\
111211112100000&000\\
000000001221111&000\\
000002000112111&000\\
000001200011211&000\\
000000120011121&000
\end{pmatrix}}\]
\end{example}

It is worth noting that $5\not\equiv 1\pmod 3$ and the code $C_{\pi}$ is not necessarily LCD. For example, if we take $C_{\pi}=\langle\begin{pmatrix} 110011\\ 001111\end{pmatrix}\rangle$, which is not an LCD code, with the same $C_{\phi}$, we get another $[18,10,4]$ LCD ternary code with a generator matrix $gen_4$. The first code contains 30 codewords of weight 4, while the second code has 40 codewords with minimum weight. This example shows that the condition on $m$ in Corollary \ref{cor:FsigQS} is important.

\section{Conclusion}

The decomposition of the linear codes that have non-trivial permutation automorphisms gives us a powerful construction for codes with optimal parameters and important properties. Many optimal self-dual codes with automorphisms of prime order over fields with 2, 3 or 4 elements are obtained using their previously known structure \cite{Huffman48,Huffman3,Huffman4II,Yorgov56}.

As a conclusion, we would like to present what is known so far about the automorphisms of the putative binary self-dual $[72,36,16]$ code.  The extremal self-dual codes of
length a multiple of $24$ are of particular interest, but only two such codes are known so far - the extended Golay code $g_{24}$ and the extended quadratic residue code $q_{48}$ (see \cite{Pless68,HLT}). In  1973 Sloane \cite{Sloane}
posed a question which remains unresolved:
is there a binary self-dual doubly-even $[72,36,16]$ code? The automorphism group of
the extended Golay code  is the
5-transitive Mathieu group $M_{24}$ of order
$2^{10}\cdot 3^3\cdot 5\cdot 7\cdot 11\cdot 23$ (see \cite{Berlekamp}), as the automorphism group of $q_{48}$ is only 2-transitive and is isomorphic to the projective special linear group $\mathrm{PSL(2,47)}$ of order
$2^4\cdot 3\cdot 23\cdot 47$ \cite{KS}. The first authors to study the automorphism group of the putative $[72,26,16]$ code were Conway and Pless \cite{CP}, in particular they focused on the possible automorphisms of odd prime order. The last published result on this group so far is given in \cite{Borello_Magdeburg} and we present it in the following theorem.

\begin{theorem}
If $C$ be self-dual $[72, 36, 16]$ code, then $\Aut(C)$ is trivial or
isomorphic to $C_2$, $C_3$, $C_2\times C_2$ or $C_5$.
\end{theorem}

So, if such a code exists, it may be very difficult to find it by algebraic techniques. Many authors have tried combinatorial methods to construct such a code, also using its connections with combinatorial designs and Hadamard matrices \cite{Assmus,BFW}, but all the resulting codes have minimum distance $d\le 12$.
Although efforts to obtain such a code by permutation of a given order have so far been unsuccessful, the method for constructing self-dual, self-orthogonal and LCD codes has been used many times and many new codes have been introduced through it. So studying the structure of linear codes with permutation automorphisms provides a powerful method to investigate, construct and classify codes with given properties and parameters.


\end{document}